\documentclass[copyright,creativecommons]{eptcs}
\providecommand{\event}{Gandalf 2018}

\newif\ifproof 
\prooffalse

\usepackage{enumerate}
\usepackage{amsmath,amsthm,accents,amssymb,amsfonts,xspace}
\usepackage{thmtools,thm-restate}
\usepackage{xcolor,pgf,tikz,wasysym}
\usetikzlibrary{arrows,automata,shapes,decorations,calc,positioning}
\usepackage{truc}
\def\epsilon{\varepsilon}

\title{Multi-weighted Markov Decision Processes\\ with Reachability Objectives}
\author{Patricia Bouyer  \hspace*{1cm}   Mauricio Gonz\'{a}lez
  \institute{LSV -- CNRS \& ENS Paris-Saclay -- Universit\'e Paris-Saclay -- France}
  \email{\{patricia.bouyer,mauricio.gonzalez\}@lsv.fr}
  \and
  Nicolas Markey
  \institute{IRISA -- CNRS \& Univ. Rennes \& INRIA -- France}
  \email{nicolas.markey@irisa.fr}
  \and
  Mickael Randour
  \institute{FNRS \& UMONS -- Belgium}
  \email{mickael.randour@umons.ac.be}
}
\def\authorrunning{P.~Bouyer, M.~Gonz\'alez, N.~Markey, M.~Randour}

\def\titlerunning{Multi-weighted MDPs with Reachability Objectives}

\begin{document}
\maketitle

\begin{abstract}
  In this paper, we are interested in the synthesis of schedulers in
  double-weighted Markov decision processes, which satisfy both a
  percentile constraint over a weighted reachability condition, and a
  quantitative constraint on the expected value of a random variable
  defined using a weighted reachability condition. This problem is
  inspired by the modelization of an electric-vehicle charging
  problem.  We study the cartography of the problem, when one
  parameter varies, and show how a partial cartography can be obtained
  via two sequences of opimization problems. We discuss completeness
  and feasability of the method.
  % We show how to compute solutions (when they exist), and discuss
  % completeness and feasibility of the method.
\end{abstract}

\section{Introduction}

\paragraph{Importing formal methods in connected fields.}
Formal methods can help providing algorithmic solutions for control
design. The~electric-vehicle (EV) charging problem is an example of
such an application area. This~problem, usually presented as a control
problem (see~eg.~\cite{BLHM16}), can~actually be modelled
% \slash abstracted
using Markov decision processes (MDP)~\cite{DI14,JP16,GBBLM17}.
Probabilities provide a way to model the non-flexible part of the
energy network (consumption outside~EV, for~which large databases
exist---and from which precise statistics can be extracted); we~can
then express an upper bound on the peak load as a safety condition
(encoded as a reachability condition in our finite-horizon model),
the~constraint on the charging of all vehicles as a quantitative
reachability objective, and various optimization criteria
(e.g.~minimizing the ageing of distribution transformers, or the
energy price) as optimization of random cost variables.

Due to the specific form of the constructed model (basically
acyclicity of the model), an~ad-hoc method could be implemented using
the tool PRISM, yielding interesting practical results as reported
in~\cite{GBBLM17} and in a forthcoming paper. However,
the~computability of an optimal strategy in a general~MDP, as~well as
the corresponding decision problem, was unexplored.
% is actually unknown.

\paragraph{Markov decision processes.}
MDPs have been studied for long~\cite{puterman94,FV97}. An~MDP is a
finite-state machine, on which a kind of game is played as
follows. In~each~state, several decisions (a.k.a.~actions) are
available, each yielding a distribution over possible successor
states.  Once an action is selected, the~next state is chosen
probabilistically, following the distribution corresponding to the
selected action. The~game proceeds that way \textit{ad~infinitum},
generating an infinite play. The~way actions are chosen is according
to a \emph{strategy} (also~called \emph{policy} in the context of
MDPs). Rewards~and\slash or costs can be associated to each action or
edge, and various rules for aggregating individual rewards and costs
encountered along a play can be applied to obtain various payoff
functions. Examples of payoff functions include:
\begin{itemize}
\item sum up all the encountered costs (or rewards) along a play,
  until reaching some target (finite if the target is reached,
  infinite otherwise)---this~is the so-called \emph{truncated-sum}
  payoff;
\item sum up all the encountered costs (or rewards) along a play with
  a discount factor at each step---this~is the so-called
  \emph{discounted-sum} payoff;
\item average over the encountered costs (or rewards) along a
  play---this~is the so-called \emph{mean-payoff}.
\end{itemize}
Those payoff functions have been extensively studied in the
literature; discounted-sum and mean-payoff have been shown to admit
optimal memoryless and deterministic strategies, which can be computed
using linear programming, yielding a polynomial-time
algorithm. Alternative methods, such as value iteration or policy
improvement, can be used in practice. On~the other~hand, the
\emph{shortest-path} problem, which aims at minimizing the
truncated-sum, has been fully understood only recently~\cite{BBD+18}:
one can decide in polynomial time as well whether the shortest-path is
finite, or~whether it is equal to $+\infty$ (if~one cannot ensure
reaching the target almost-surely), or whether it can be smaller than
any arbitrary small number (if~a~negative loop can be enforced---note
that in the context of stochastic systems, such a statement may be
misleading, but it corresponds to a rough intuition), and
corresponding strategies can be computed.

\paragraph{Multi-constrained problems in MDPs.}
The paradigm of multi-constrained objectives in stochastic systems in
general, and in MDPs in particular, has recently arisen. It~allows to
express various (quantitative or qualitative) properties over the
model, and to synthesize strategies accordingly. This~new field of
research is very rich and ambitious, with~various types of objective
combinations (see~for~instance~\cite{BDK14,RRS15b} for~recent
overviews). For recent developments on MDPs, one can cite:
\begin{itemize}
\item Pareto curves, or percentile queries, of multiple quantitative
  objectives: given several payoff functions, evaluate
  which tradeoff can be made between the probabilities, or the
  expectations, of~the various payoff
  functions. In~\cite{chatterjee07b,BBC+14,CKK17}, solutions based on
  linear programming are provided for mean-payoff
  objectives. The~percentile-query problem for
  various quantitative payoff functions is studied in~\cite{RRS17}.
\item probability of conjunctions of objectives: given several payoff
  functions, evaluate the probability that all
  constraints are satisfied. This problem is studied in~\cite{HK15} for
  reachability (that is, for the truncated-sum payoff function),
  a \PSPACE lower bound is proved for that
  problem, already with a single payoff function.
\item the ``beyond worst-case'' paradigm: satisfy
  both a safety constraint on all outcomes, and various performance
  criteria. Variations of this problem for various
  payoff functions have been studies in~\cite{CR15,BFRR17,BRR17}.
\item conditional expectations~\cite{BKKW17} or conditional
  values-at-risk~\cite{KM18}, which measures likelihoods of properties
  under some assumptions on the system, have  recently been
  investigated.
\end{itemize}

\paragraph{Our contributions.}
% We model the problem arising from the EV-charging problem as
% modelled in~\cite{GBBLM17} into a multi-constrained problem
% over~MDPs.
The~general multi-constrained problem, arising from the EV-charging
problem as modelled in~\cite{GBBLM17}, takes as an input an MDP with
two weights, $w_1$~and~$w_2$, and~requires the existence
(and~synthesis) of~a~strategy ensuring that some (absorbing) target
state be reached, with a percentile constraint on the truncated sum
of~$w_1$ (lower bound parameterized by~$\epsilon$), and an expectation
constraint on the truncated sum of~$w_2$. The~initial EV-charging
problem corresponds to the instance of that problem when
${\epsilon=0}$, where $w_1$ represents the energy that is used for
charging and $w_2$ represents the ageing of the transformer.

As defined above, our problem integrates both the ``beyond
worst-case'' paradigm of~\cite{BFRR17}, and percentile queries as
in~\cite{CKK17} (mixing probabilities and expectations).  While
in~\cite{CKK17} linear programs are used for solving percentile
queries (heavily relying on the fact that mean-payoff objectives are
tail objectives), we need different techniques since the truncated-sum
payoff is very much prefix-dependent; actually, the \PSPACE-lower
bound of~\cite{HK15} immediately applies here as~well (even~without a
constraint on the expectation of~$w_2$). We~develop here a methodology
to describe the cartography of the problem, that is, the~set of values
of the parameter~$\epsilon$ for which the problem has a solution.
% It uses two sequences of optimization problems, allowing to
% partially characterize the cartography.
Our~approach is based on two sequences of optimization problems which,
in some cases we characterize, allow to have the (almost) full
picture. We~then discuss computability issues.
%
%By~lack of space, some proofs could not be included in the paper, and
%are given in the eponymous research report. {\color{red}\fbox{OK??}}

\section{Preliminary definitions}

Let $S$ be a finite set. We write $\Dist(S)$ for the set of
distributions over~$S$, that~is, the set of functions
$\delta \colon S \to [0,1]$ such that $\sum_{s \in S} \delta(s)=1$.
A~distribution over~$S$ is \emph{Dirac} if $\delta(s)=1$ for
some~$s\in S$.

\subsection{Definition of the model}

In this paper, we mainly focus on doubly-weighted Markov decision
processes, but the~technical developments mainly rely on simply-weighted
Markov decision processes. We~therefore define the setting
with an arbitrary number of weights.

\begin{definition}
  Let $k \in \mathbb{N}$.  A~\emph{$k$-weighted Markov decision
    process} (\emph{$k$w-MDP}) is a tuple $\calM =
  (S,s_\init,\Goal,\penalty1000 E,(w_i)_{1 \le i \le k})$, where:
  \begin{itemize}
  \item $S$ is a finite set of states;
  \item $s_\init \in S$ is the initial state;
  \item $\Goal \in S$ is the target state;
  \item $E \subseteq S \times \Dist(S)$ is a finite set of stochastic
    edges;
  \item for each $1 \le i \le k$, the function $w_i \colon S \times S
    \to \mathbb{Q}$ assigns a rational weight to
    each transition of the complete graph with state space~$S$. 
  \end{itemize}
\end{definition}
A (finite, infinite) path in~$\calM$ from~$s$ is a (finite, infinite)
sequence of states $s_0 s_1 s_2 \dots$ such that $s_0=s$ and for
every~$i$, there is $\delta_i \in \Dist(S)$ such that $(s_i,\delta_i)
\in E$ and $\delta_i(s_{i+1})>0$. Finite paths are equivalently called
histories. We~write~$\Paths^\calM(s)$ (resp.~$\Paths^\calM_\infty(s)$)
for the set of paths (resp.~infinite paths), in~$\calM$ from
state~$s$. Given a history $h=s_0 s_1 s_2 \dots s_N$ and $\ell\leq N$,
the~$w_i$-accumulated weight of~$h$ after $\ell$ steps is defined~as
$\Acc_{w_i}^\ell(h) = \sum_{j=1}^\ell w_i(s_{j-1},s_j)$. This notion
extends straightforwardly to infinite paths.

A (randomized) \emph{strategy} in~$\calM$ is a function $\sigma$
assigning to every history $h = s_0 s_1 s_2 \dots s_N$ a distribution
over $s_N E = \{\delta \in\Dist(S) \mid (s_N,\delta) \in
E\}$. A~strategy~$\sigma$~is said to be \emph{pure} whenever the
distributions it prescribes are Dirac.
A~path~$s_0 s_1 s_2 \dots$ is an outcome of~$\sigma$ whenever for
every strict prefix $s_0 s_1 s_2 \dots s_N$, there exists $\delta \in s_N
E$ such that $\sigma(h)(\delta)>0$ and $\delta(s_{N+1})>0$. Basically,
the~outcomes of a strategy are the paths that are activated by the
strategy. We~write $\out^\calM(\sigma,s)$
(resp.~$\out_\infty^\calM(\sigma,s)$) for the set of finite
(resp. infinite) outcomes of~$\sigma$ from state~$s$.

Given a strategy $\sigma$ and a state~$s$, we~denote with
$\Prob_{\sigma,s}^\calM$ the probability distribution, according to
$\sigma$, over the infinite paths in~$\Paths^\calM_\infty(s)$, defined
in the standard way using cylinders based on finite paths from~$s$.
If $f$ is a measurable functions from $\Prob_{\sigma,s}^\calM$ to
$\bbR$, we denote by $\Esp_{\sigma,s}^\calM(f)$ the expected value of
$f$ w.r.t. the probability distribution $\Prob_{\sigma,s}^\calM$, that
is, $\Esp_{\sigma,s}^\calM(f) = \displaystyle \int f ~\textrm{d} \Prob_{\sigma,s}^\calM$.
In~all notations, we may omit to mention the superscript~$\calM$ when
it is clear in the context, and may omit to mention the starting
state~$s$ when it is~$s_\init$, so that $\Prob_\sigma$ corresponds to
$\Prob_{\sigma,s_{\init}}^\calM$.

\begin{example}
  Consider the 2w-MDP $\calM$ of Figure~\ref{fig:ex}. It~has four
  states~$s_0$, $s_1$, $s_2$ and~$\Goal$, and five edges, labelled
  with their names (here~$a$, $b$, $c$, $d$ and~$e$). Weights label
  pairs of states (but are represented here only for pairs of states
  that may be activated).  Edges~$a$, $b$, $d$ and~$e$ have Dirac
  distributions, while edge~$c$ has a stochastic choice (represented
  by the small black square).  For~readability we do not write the
  exact distributions, but in this example, they are assumed to be
  uniform.
  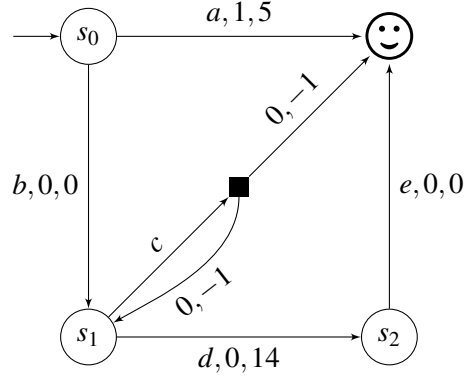
\begin{figure}[h]
    \centering
    \begin{tikzpicture}
      \draw (0,0) node [circle,draw] (s0) {$s_0$};
      \draw (4,0) node [circle,inner sep=-1mm] (t) {\scalebox{2.5}{$\Goal$}};
      \draw (0,-4) node [circle,draw] (s1) {$s_1$};
      \draw (4,-4) node [circle,draw] (s2) {$s_2$};
      
      \draw (2,-2) node [fill=black] (rand) {};

      \draw [-latex'] (rand) -- (t) node [midway,above,sloped]
      {$0,-1$};
      \draw [-latex'] (s0) -- (t) node [midway,above,sloped]
      {$a,1,5$};
      \draw [-latex'] (s0) -- (s1) node [midway,left] {$b,0,0$};
      \draw [-latex'] (s1) -- (rand) node[midway,above,sloped] {$c$};
      \draw [-latex'] (rand) .. controls +(-90:1cm) and +(30:1cm)
      .. (s1) node [midway,sloped,below] {$0,-1$};
      \draw [-latex'] (s1) -- (s2) node [midway,below] {$d,0,14$};
      \draw [-latex'] (s2) -- (t) node[midway,right]{$e,0,0$};

      \draw [latex'-] (s0) -- +(-1,0);
    \end{tikzpicture}
    \caption{An example of a 2w-MDP}
    \label{fig:ex}
  \end{figure}
\end{example}

\subsection{Payoff functions}

We are interested in quantitative reachability properties (also called
truncated-sum in the literature), which we formalize as follows.  Let
$\rho = s_0 s_1 s_2 \dots \in \Paths(s_0)$. We~use standard \LTL-based
notations for properties; for~instance, we~write
$\rho \models \F \Goal$ (resp.~$\rho \models \F[I] \Goal$, when $I$ is
an interval of $\bbN$) when there is $j$ (resp. $j \in I$) such that
$s_j = \Goal$, and $\rho \models \G \neg\Goal$ (resp.
$\rho \models \G[I] \neg\Goal$, when $I$ is an interval of $\bbN$)
when $s_j \ne \Goal$ for every~$j$ (resp.~for every~$j \in
I$). We will often use expressions $\sim N$ (in $\{<,\le,=,\ge,>\}
\times \bbN$ for defining intervals of $\bbN$.

If~$\rho\models \F\Goal$ and $1\leq i\leq n$, we define the $i$-th
payoff function $\TS^\Goal_{w_i}(\rho)$ as $\Acc_{w_i}^N(\rho)$ where
$N$ is the least index such that $\rho \models \F[=N]
\Goal$. If~$\rho\not\models \F\Goal$ then
$\TS^\Goal_{w_i}(\rho)=+\infty$.
The~function~$\rho \mapsto \TS^\Goal_{w_i}(\rho)$ is measurable, hence
for every $\bowtie \nu$ in
$\{\mathord<,\mathord\leq,\mathord=,
\mathord\geq,\mathord>\}\times\mathbb Q$,
$\Prob_{\sigma,s_0}(\{\rho \in \Paths_\infty(s_0) \mid
\TS_{w_i}^\Goal(\rho) \bowtie \nu \})$ (simply written as
$\Prob_{\sigma,s_0}(\TS_{w_i}^\Goal \bowtie \nu)$) and
$\Esp_{\sigma,s_0}(\TS_{w_i}^\Goal)$ are well-defined.  We~write
$\rho \models (\TS_{w_i}^\Goal \bowtie \nu)$ whenever
$\TS_{w_i}^\Goal(\rho) \bowtie \nu$.

\medskip
In the rest of the paper, we assume that $\Goal$ is a sink state,
and that there is a single loop on~$\Goal$ whose weights are all
equal to~$0$. This is w.l.o.g. since we will study payoff functions
$\TS^\Goal_{w_i}$, which only consider the prefix up to the first
visit to~$\Goal$.

\begin{example}
  Consider again the example of Figure~\ref{fig:ex}. Consider the
  strategy~$\sigma$ which selects $a$ or $b$ uniformly at random in
  $s_0$, and always selects~$c$ in~$s_1$. Then,
  \begin{xalignat*}3
    \Prob_{\sigma,s_0} \Big( \F \Goal \Big) &= 1\quad
    &
    \Prob_{\sigma,s_0} \Big( \TS_{w_1}^\Goal \ge 1 \Big) & = 
    \frac{1}{2}\quad
    &
    \Esp_{\sigma,s_0} \Big(\TS_{w_2}^\Goal\Big) &= \frac{1}{2}
    \cdot 5 + \frac{1}{2} \cdot \sum_{i=1}^\infty \frac{-i}{2^i}  = 1+\frac{1}{2}
  \end{xalignat*}
\end{example}

\section{The problem}

The problem that we tackle in this paper arises from a recent study
of an EV-charging problem~\cite{GBBLM17}. The~general
problem we will define is a relaxed version of the original problem,
combining several stochastic constraints (a~percentile query over some
payoff function and a constraint on the expectation of some payoff
function) together with a worst-case obligation. While
various payoff functions could be relevant, we focus on those payoff
functions that were used for the EV-charging problem, that is,
quantitative reachability (i.e., the truncated-sum payoff).  We will
see that the developed techniques are really specific to our choice of
payoff functions.

In this paper, we focus on a combination of \emph{sure reachability}
of the goal state, of a percentile constraint on the proportion of
paths having high value for the first payoff, and of a constraint on
the expected value of the second payoff. 

Let $\calM = (S,s_\init,\Goal,E,(w_1,w_2))$ be a
2w-MDP. Let $\nu_1,\nu_2 \in \mathbb{Q}$. For~every~$\epsilon \ge 0$,
we~define the problem $\problem_{\calM,\nu_1,\nu_2}(\epsilon)$ as
follows: there exists a strategy $\sigma_\epsilon$ such that
\begin{enumerate}
\item for all $\rho \in \out^\calM_{\infty}(\sigma_\epsilon,s_\init)$,
  it~holds $\rho \models \F \Goal$;
\item $\Prob^\calM_{\sigma_\epsilon,s_\init}\Big(\TS_{w_1}^{\Goal} \ge
  \nu_1\Big) \ge 1-\epsilon$;
\item $\Esp^\calM_{\sigma_\epsilon,s_\init}\Big(\TS_{w_2}^\Goal\Big) <
  \nu_2$.
\end{enumerate}
We aim at computing the values of $\epsilon$ for which
$\problem_{\calM,\nu_1,\nu_2}(\epsilon)$ has a solution.  For~the rest
of this section, we~assume that there is a strategy~$\sigma$ such that
${\Esp_{\sigma,s_\init}\Big(\TS_{w_2}^\Goal\Big) < \nu_2}$. Otherwise
$\problem_{\calM,\nu_1,\nu_2}(\epsilon)$ trivially has no solutions,
for~any~$\epsilon$. This can be decided using the algorithm recently
developed in~\cite{BBD+18}.

\begin{example}
  To illustrate the problem, we consider again the example given in
  Figure~\ref{fig:ex}. Consider $\epsilon = 0.5$, $\nu_1=1$ and
  $\nu_2=4.3$.  The only way to satisfy the threshold constraint
  on~$w_1$ is that at~least half of the paths use~$a$, impacting $2.5$
  over the expectation of~$w_2$. The~other paths have to go to~$s_1$,
  and then take $c$ for some time (provided the play goes back
  to~$s_1$) in~order to decrease the expectation of~$w_2$, before it
  becomes possible to take $d$ and then~$e$ (so~that the strategy is
  surely winning). This~strategy uses both randomization (at~$s_0$)
  and memory (counting the number of times $c$~is taken before $d$ can
  be taken).
\end{example}

We~call the \emph{cartography} of our problem the function which
associates to every $\epsilon \in [0;1]$, either \texttt{true} if
$\problem_{\calM,\nu_1,\nu_2}(\epsilon)$ has a solution, or
\texttt{false} otherwise. It~is easily seen that the cartography is a
threshold function, and can be characterized by an
interval~$I=\langle\gamma;1]$ (which may be left-open or left-closed): 
In~what follows, we~describe an algorithmic technique to approximate
this interval, and under additional conditions, to compute the
bound~$\gamma$. Whether the bound belongs to the interval
remains open in general.

%{%\color{red}
\paragraph{Link with the electric-vehicle (EV) charging problem.}
The (centralized) EV-charging problem consists in scheduling power
loads within a time interval $[0;T]$ ($T$~being a fixed time bound)
with uncertain exogenous loads, so as to minimize the impact of
loading on the electric distribution network (measured through the ageing of
the  transformer, which depends on the temperature of its winding).
Following standard models, time is discretized,
and the instantaneous energy consumption at time~$t$ can be written as
the sum of the non-flexible load $\ell_{t}^{\nonflex}$ (consumption
outside~EV) and the flexible load $\ell_t^{\flex}$, corresponding to
the EV charging.  The~flexible loads at each time are controllable actions,
while the non-flexible part is known, or
statistically estimated using past databases.

A first constraint on the transformer is given by its capacity:
$\ell_t^{\flex} + \ell_{t}^{\nonflex} \le \textsf{L}^{\max}$ (where
$\textsf{L}^{\max}$ is a constant) for every $0 \le t \le T$.
A~second constraint represents the charge required for charging
all vehicles on schedule:
$\sum_{t=0}^T\ell_{t}^{\flex}\geq \textsf{LoC}^{\max}$, where
$\textsf{LoC}^{\max}$ is a constant.
%, which represents the fact that
%the vehicles have been charged within the interval $[0;T]$.
The~flexible load~$\ell_{t}^{\flex}$ can thus be seen as a weight function~$w_1$.

While greedy solutions can be used to solve the above constraints,
the~ageing of thansformer has not been taken into account so~far.
Using a standard model for the ageing of a transformer
(see~\cite{BLHM16,GBBLM17} for~details), it~can be expressed as a
weight function~$w_2$ based on a discrete model in which states
aggregate information on the system at the the two last
timepoints. Globally, a~2w-MDP~$\calM$
% (with appropriate states)
can be built, such that a controller for the EV-charging problem
coincides with a solution to
$\problem_{\calM, \textsf{LoC}^{\max},\nu_2}(0)$, for some bound
$\nu_2$ for the expected ageing of the transformer.
%}

\iffalse

The ageing of the transformer is modelled by the equation (see
\cite{At}) $A_t := e^{\alpha \theta^{\HS}_t+\beta}$ where
$\alpha>0>\beta$ are constants, and where $\theta^{\HS}_t$ is the
hot-spot of the transformer at time $t$. The temperature
$\theta^{\HS}_t$ depends on its value at time $t-1$, the ambient
temperature at time $t$, and consumptions at time $t$ and $t-1$,
$\ell_t$ and $\ell_{t-1}$, as follows:
\[
\theta^{\HS}_{t} :=  a\cdot\theta^{\HS}_{t-1} + b_1\cdot\ell_{t}^2 + b_2\cdot\ell_{t-1}^2 + c_{t}
\]
where $a\in[0,1]$, $b_1\geq 0\geq b_2$ are constants, and $c_{t}$ is a
known function of the ambient temperature, which ensures that
$\theta^{\HS}_{t}\geq 0$ for all $t$. A constraint on the hot-spot if
given by $\theta^{\HS}_t\leq \theta^{\max}$, where $\theta^{\max}$ is
a constant.

The cost of electricity for a user is then
supposed to be a linear combination of the ageimg of the transformer
and of the price to pay, that is globally:
\[
  \textsf{C} := \sum_{t=0}^T\left(\lambda\cdot A_t+(1-\lambda)\cdot
    p_t \cdot \ell_t\right)
\]

The EV-charging problem then consists in chosing appropriate values
for the $(\ell_t^{\flex})_{t=0}^T$ for minimizing \textsf{C} under the
various constraints that we have seen. In particular, the target
$\sum_t \ell^\flex_{i,t} \ge \textsf{LoC}_i^{\max}$ has to be reached.

\fi

\section{Approximated cartography}

We fix a 2w-MDP $\calM = (S,s_0,\Goal,E,w_1,w_2)$ and two thresholds
$\nu_1,\nu_2 \in \mathbb{Q}$.  We~introduce two simpler optimization
problems related to $\problem_{\calM,\nu_1,\nu_2}(\epsilon)$, from
which we derive informations on the good values of~$\epsilon$ for
which that problem has a solution.  As~we explain below, our~approach
is in general not complete. However, we~observe that the true part of
the~cartography of our problem is an interval of the
form~$\langle \gamma; 1]$; under some hypotheses, we~prove that our
approach allows to approximate arbitrarily 
% algorithm computes the exact
the bound~$\gamma$, but may not be able to decide if the interval is
left-open or left-closed.

\subsection{Optimization problems}

Let $N$ be an integer.  We~write~$\phi_N^+$ for the property $\F[\le
  N] \Goal \wedge \TS_{w_1}^\Goal \ge \nu_1$ (which specifies that the
target is reached in no more than $N$~steps, with a $w_1$-weight
larger than or equal to~$\nu_1$), and $\phi_N^-$ for the property
$\F[\le N] \Goal \wedge \TS_{w_1}^\Goal < \nu_1$ (which means that the
target is reached in no more than $N$~steps, with a $w_1$-weight
smaller than~$\nu_1$). We write $\psi_N$ for the property $\G[\le N]
\neg \Goal$ (the~target is not reached during the $N$ first steps).
By~extension, we write $\phi^+$, $\phi^-$ and $\psi$ for the
properties $\F \Goal \wedge \TS_{w_1}^\Goal \ge \nu_1$, $\F \Goal
\wedge \TS_{w_1}^\Goal < \nu_1$ and $\G \neg \Goal$.
Finally, we~may also (abusively) use such formulas to denote
the set of paths that satisfy~them.

For~every~$N$ and every path~$\rho$ of~$\calM$ of length at least~$N$,
it~holds that: \( \rho \models \phi_N^+ \vee \phi_N^- \vee \psi_N \).
Moreover, observe that $\phi_N^+\subseteq \phi^+$ and
$\phi^+\subseteq \phi_N^+\vee \psi_N$.  As~a consequence,
for~every~$N$ and every strategy~$\sigma$,
\( \Prob_\sigma(\phi_N^+) \le \Prob_\sigma(\phi^+) \le
\Prob_\sigma(\phi_N^+\vee \psi_N) \).

\subsubsection{First optimization problem}

We~define
\[
  \overline{\val}_N = \inf \Big\{\Prob_\sigma \Big(\phi_N^- \vee
  \psi_N\Big) \mid \sigma\ \text{s.t.}\
  \Esp_{\sigma}\Big(\TS_{w_2}^\Goal\Big) < \nu_2\Big\}
\]
and for every $\alpha>0$, we~fix a witnessing
strategy~$\sigma_{N,\alpha}$ for $\overline{\val}_N$ up to~$\alpha$ (i.e.
$\Esp_{\sigma_{N,\alpha}}(\TS_{w_2}^\Goal) < \nu_2$ and
$\Prob_{\sigma_{N,\alpha}} (\phi_N^- \vee \psi_N) \leq \overline{\val}_N+\alpha$).

\begin{remark}
  Note that, since we assume that there is a strategy $\sigma$ such
  that $\Esp_{\sigma}(\TS_{w_2}^\Goal) < \nu_2$, the
  constraint of this optimization problem is non-empty.
  Note also that if $\sigma$ is a strategy such that
  $\Esp_{\sigma}(\TS_{w_2}^\Goal) < \nu_2$, then
  $\Prob_\sigma( \F \Goal) = 1$, since for every
  path~$\rho$, $\TS_{w_2}^\Goal(\rho) = +\infty$ whenever $\rho
  \not\models \F \Goal$.
\end{remark}

\ifproof
\begin{restatable}{lemma}{lemmatwo}
For every $N$, it~holds $\overline\val_{N+1} \le \overline\val_N$.
\end{restatable}
\fi

\ifproof
\begin{proof}
  Fix some $\alpha>0$, and consider strategy~$\sigma_{N,\alpha}$.
  Clearly, 
  \(
  \phi_{N+1}^- \subseteq \psi_N \vee \phi_N^-
  \),
  and
  \(
  \psi_{N+1} \subseteq \psi_N
  \).
  Hence
  \begin{equation*}
    \overline\val_{N+1} \leq \Prob_{\sigma_{N,\alpha}}
    \Big(\phi_{N+1}^- \vee 
    \psi_{N+1}\Big) \leq \Prob_{\sigma_{N,\alpha}} 
    \Big(\phi_{N}^- \vee \psi_{N}\Big)
    \leq \overline\val_N+\alpha.
  \end{equation*}
  This allows to conclude that $\overline\val_{N+1} \le
  \overline\val_N$.
\end{proof}
\fi

\ifproof\else
It is not hard to see that the sequence
$(\overline\val_{N})_{N\in\bbN}$ is non-increasing (see~Appendix).
\fi
We let $\overline\gamma =
\lim_{N \to +\infty} \overline\val_{N}$.  We~then have:

\begin{restatable}{lemma}{lemmathree}
\label{lemma-eps<gamma}
  For every $\epsilon < \overline\gamma$,
  $\problem_{\calM,\nu_1,\nu_2}(\epsilon)$ has no solution.
\end{restatable}

\begin{proof}
  Fix $\epsilon < \overline\gamma$, and assume towards a contradiction
  that $\problem_{\calM,\nu_1,\nu_2}(\epsilon)$ has a solution.

  Fix a winning strategy~$\sigma_\epsilon$  for
  $\problem_{\calM,\nu_1,\nu_2}(\epsilon)$. By~the first winning
  constraint, there exists $N_\epsilon$ such that any outcome
  $\rho \in \out_\infty(\sigma_\epsilon)$ satisfies $\F[\le N_\epsilon]
  \Goal$ (thanks to K\"onig's lemma). Furthermore, since
  $\Esp_{\sigma_\epsilon}(\TS_{w_2}^\Goal) < \nu_2$, then
  $\sigma_\epsilon$ belongs to the domain
  of the optimization problem defining~$\overline\val_{N_\epsilon}$.
  Hence, we have 
  \[
    \overline\val_{N_\epsilon} \le  
    \Prob_{\sigma_\epsilon}\Big(\phi_{N_\epsilon}^- \vee
    \psi_{N_\epsilon}\Big) 
     =  1-
    \Prob_{\sigma_\epsilon}\Big(\phi_{N_\epsilon}^+\Big) 
     =  1-
    \Prob_{\sigma_\epsilon}\Big(\phi^+\Big)
     \le  \epsilon.
  \]
  This is then a contradiction with $\epsilon < \overline\gamma \le
  \overline\val_{N_\epsilon}$. Hence, we deduce that for every
  $\epsilon < \overline\gamma$,
  $\problem_{\calM,\nu_1,\nu_2}(\epsilon)$ has no solution.
\end{proof}

\begin{restatable}{lemma}{lemmafour}
\label{lemma:sol}
  For every $N$, for every $\epsilon>\overline\val_N\geq\overline\gamma$,
  $\problem_{\calM,\nu_1,\nu_2}(\epsilon)$ has a solution.
\end{restatable}

\begin{proof}
  Let $N$ be an integer, and $\epsilon>\overline\val_N$. Let
  $\sigma_N$ be a strategy such that
  \[
  \Prob_{\sigma_N} \Big(\phi_{N}^- \vee \psi_{N}\Big) <
  \epsilon\ \text{and}\
  \Esp_{\sigma_N}\Big(\TS_{w_2}^\Goal\Big) < \nu_2
  \]

  Let $\sigma_{\text{Att}}$ be an attractor (memoryless) strategy on
  $\calM$, that is, a strategy which enforces reaching~$\Goal$~; write
  $M$ for a positive upper-bound on the accumulated weight~$w_2$ when
  playing that strategy (from any state).  For $k \ge N$, define
  $\sigma_N^k$ as: play $\sigma_N$ for the first $k$~steps, and if
  $\Goal$ is not reached, then play $\sigma_{\text{Att}}$. We show
  that we can find $k$ large enough such that this strategy is a
  solution to $\problem_{\calM,\nu_1,\nu_2} (\epsilon)$.

  The first condition is satisfied, since either the target state is
  reached during the $k$ first steps (i.e. while playing $\sigma_N$),
  or it will be surely reached by playing $\sigma_{\text{Att}}$. Since
  $\Prob_{\sigma_N^k} (\phi_{N}^+) = 1 - \Prob_{\sigma_N^k}
  (\phi_{N}^- \vee \psi_{N})$ and
  $\Prob_{\sigma_N^k} (\phi_{N}^+) \le \Prob_{\sigma_N^k} (\phi^+)$,
  it is the case that $\Prob_{\sigma_N^k} (\phi^+) \ge 1-\epsilon$,
  which is the second condition for being a solution to
  $\problem_{\calM,\nu_1,\nu_2}(\epsilon)$. Finally, thanks to the law
  of total expectation, we can write:
  \begin{xalignat*}1 \Esp_{\sigma_N^k}\Big(\TS_{w_2}^\Goal\Big) & =
    \Esp_{\sigma_N^k}\Big(\TS_{w_2}^\Goal \mid \F[\le k] \Goal\Big)
    \cdot \Prob_{\sigma_N^k}(\F[\le k] \Goal) +
    \Esp_{\sigma_N^k}\Big(\TS_{w_2}^\Goal \mid \G[\le k] \neg
    \Goal\Big) \cdot \Prob_{\sigma_N^k}(\G[\le k]
    \neg\Goal) \\
    \noalign{\allowbreak} & \le \Esp_{\sigma_N^k}\Big(\Acc^k_{w_2}
    \mid \F[\le k] \Goal\Big) \cdot \Prob_{\sigma_N^k}(\F[\le k]
    \Goal) + \Esp_{\sigma_N^k}\Big(\Acc^k_{w_2} + M \mid \G[\le k]
    \neg \Goal\Big) \cdot \Prob_{\sigma_N^k}(\G[\le k]
    \neg\Goal) \\
    & \tagcomment{since the global impact of playing the strategy
      $\sigma_{\text{Att}}$ is bounded by $M$} \\
    \noalign{\allowbreak} & = \Esp_{\sigma_N^k}\Big(\Acc^k_{w_2}\Big)
    + M \cdot \Prob_{\sigma_N^k}(\G[\le k]
    \neg\Goal) \\
    & \tagcomment{by linearity of expectation and the law of total
      expectation again} \\
    \noalign{\allowbreak} & = \Esp_{\sigma_N}\Big(\Acc^k_{w_2}\Big) +
    M \cdot \Prob_{\sigma_N}(\G[\le k]
    \neg\Goal) \\
    & \tagcomment{since $\sigma_N^k$ coincides with
      $\sigma_N$ on the $k$ first steps}
  \end{xalignat*}
  Now, since $\Esp_{\sigma_N}(\TS_{w_2}^\Goal)$ is
  finite, it is the case that $\Prob_{\sigma_N} (\F
  \Goal) = 1$. Hence:
  \begin{itemize}
  \item $\lim_{k \to +\infty}
    \Esp_{\sigma_N}(\Acc^k_{w_2}) =
    \Esp_{\sigma_N}(\TS_{w_2}^\Goal)$, and
  \item $\lim_{k \to +\infty} \Prob_{\sigma_N}(\G[\le k]
    \neg\Goal) =0$.
  \end{itemize}
  Let $\eta =
  \nu_2-\Esp_{\sigma_N}(\TS_{w_2}^\Goal)>0$. One can
  choose $k$ large enough such that
  \[
  \Bigl|\Esp_{\sigma_N}\Bigl(\Acc^k_{w_2}\Bigr) -
  \Esp_{\sigma_N}\Bigl(\TS_{w_2}^\Goal\Bigr)\Bigr| < \eta/2
  \quad\text{ and }\quad
  \Prob_{\sigma_N}\Bigl(\G[\le k] \neg\Goal\Bigr) <
  \eta/2M.
  \]
  We conclude that:
  \[
    \Esp_{\sigma_N^k}\Big(\TS_{w_2}^\Goal\Big) < 
    \Esp_{\sigma_N}\Big(\TS_{w_2}^\Goal\Big) + \eta 
     <  \nu_2.
  \]
  The~strategy~$\sigma_N^k$ therefore  witnesses the fact that problem
  $\problem_{\calM,\nu_1,\nu_2}(\epsilon)$ has a solution.
\end{proof}

\subsubsection{Second optimization problem}

We now define
\[
\underline{\val}_N = 
\inf \Big\{\Prob_\sigma \Big(\phi_N^-\Big) \mid \sigma\
\text{s.t.}\ \Esp_{\sigma}\Big(\TS_{w_2}^\Goal\Big) <
\nu_2\Big\}.
\]
Notice that for any~$N$, $\underline\val_N\leq \overline\val_N$.
For every $\alpha>0$, we fix a witness
strategy~$\widetilde\sigma_{N,\alpha}$ for $\underline{\val}_N$ up
to~$\alpha$ (so~that
${\Esp_{\widetilde\sigma_{N,\alpha}}(\TS_{w_2}^\Goal) < \nu_2}$ and
$\Prob_{\widetilde\sigma_{N,\alpha}} (\phi_N^-) \leq
\underline{\val}_N+\alpha$).

\ifproof
\begin{lemma}
  For every $N$, we~have $\underline\val_{N} \le \underline\val_{N+1}$.
\end{lemma}

\begin{proof}
  Fix some $\alpha>0$, and consider strategy
  $\widetilde\sigma_{N+1,\alpha}$.
  It~is clear that
  \(
  \phi_{N}^- \subseteq \phi_{N+1}^-
  \),
  hence
  \[
    \underline\val_{N} \le  \Prob_{\widetilde\sigma_{N+1,\alpha}}
    \Big(\phi_{N}^- \Big)  \le \Prob_{\widetilde\sigma_{N+1,\alpha}}
    \Big(\phi_{N+1}^- \Big) \leq  \underline{\val}_{N+1}+\alpha.
  \]
  We conclude that $\underline\val_{N} \le
  \underline\val_{N+1}$.
\end{proof}
\fi

\ifproof\else
This time, it can be observed that the sequence
$(\underline\val_N)_{N\in\bbN}$ is non-decreasing.
\fi
We let $\underline\gamma = \lim_{N \to +\infty} \underline\val_{N}$.
From the results and remarks above, we~have
$\underline\val_N\leq \underline\gamma\leq\overline\gamma$ for any~$N$.
From~Lemma~\ref{lemma-eps<gamma}, we~get:

\begin{restatable}{lemma}{lemmaten}
\label{lemma:underline}
  For any~$N$ and any $\epsilon<\underline\val_N\leq\underline\gamma$,
  $\problem_{\calM,\nu_1,\nu_2}(\epsilon)$ has no solution.
\end{restatable}

While the~status of $problem_{\calM,\nu_1,\nu_2}(\overline\gamma)$ is
in general unknown, we~still have the following properties:
\begin{restatable}{prop}{laseuleprop}
  \begin{itemize}
  \item If $\problem_{\calM,\nu_1,\nu_2}(\overline\gamma)$ has a
    solution, then the sequence $(\overline\val_N)_N$ is stationary
    and ultimately takes value~$\overline\gamma$. The converse need
    not hold in general;
  \item  $\underline\gamma = \overline\gamma$ does neither imply that
    $\problem_{\calM,\nu_1,\nu_2}(\overline\gamma)$ has a solution,
    nor that $\problem_{\calM,\nu_1,\nu_2}(\overline\gamma)$ has no solution.
  \end{itemize}
\end{restatable}

\subsubsection{Summary}

Figure~\ref{fig-carto} summarizes the previous analysis.
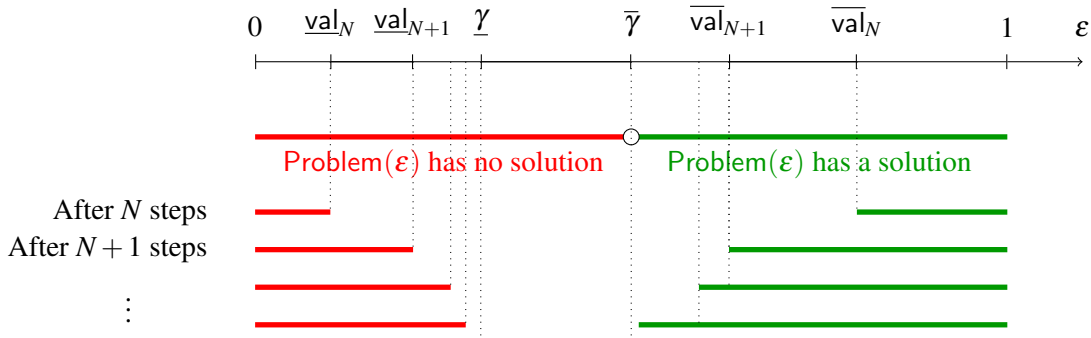
\begin{figure}[htb]
\centering
  \begin{tikzpicture}

    \path [use as bounding box] (10,1) -- (-2,-4);
    \draw [|-|] (0,0) node [above=2mm] {$0$} -- (10,0) node [above=2mm]
    {$1$};
    \draw [|-|] (3,0) node [above=2mm] {$\underline\gamma$} -- (5,0)
    node [above=2mm] {$\overline\gamma$};
    \draw [|-|] (1,0) node [above=2mm] {$\underline\val_N$} -- (2.1,0)
    node [above=2mm] {$\underline\val_{N+1}$};
    \draw [|-|] (6.3,0) node [above=2mm] {$\overline\val_{N+1}$} --
    (8,0) node [above=2mm] {$\overline\val_N$};
    \draw [->] (10,0) -- (11,0) node [above=2mm] {$\epsilon$};

    %% Theory
    
    \draw [green!60!black,line width=2pt] (5,-1) -- (10,-1) node [midway,below]
    {$\problem(\epsilon)$ has a solution};
    \draw [red,line width=2pt] (0,-1) -- (5,-1)  node [midway,below]
    {$\problem(\epsilon)$ has no solution};
    \draw [fill=white] (5,-1) circle (3pt);

    %% Practice N steps
    
    \path (-.5,-2) node [left] {After $N$ steps};
    \draw [green!60!black,line width=2pt] (8,-2) -- (10,-2);
    \draw [red,line width=2pt] (0,-2) -- (1,-2);
    % \draw [fill=white] (5,-1) circle (3pt);

    %% Practice N+1 steps
    
    \path (-.5,-2.5) node [left] {After $N+1$ steps};
    \draw [green!60!black,line width=2pt] (6.3,-2.5) -- (10,-2.5);
    \draw [red,line width=2pt] (0,-2.5) -- (2.1,-2.5);
    % \draw [fill=white] (5,-1) circle (3pt);

    \draw [green!60!black,line width=2pt] (5.9,-3) -- (10,-3);
    \draw [red,line width=2pt] (0,-3) -- (2.6,-3);

    \draw [green!60!black,line width=2pt] (5.1,-3.5) -- (10,-3.5);
    \draw [red,line width=2pt] (0,-3.5) -- (2.8,-3.5);

    \draw [dotted] (5,0) -- (5,-3.7);
    \draw [dotted] (3,0) -- (3,-3.7);

    \draw [dotted] (2.1,0) -- (2.1,-2.5);
    \draw [dotted] (1,0) -- (1,-2);
    \draw [dotted] (2.6,0) -- (2.6,-3);
    \draw [dotted] (2.8,0) -- (2.8,-3.5);
    
    \draw [dotted] (6.3,0) -- (6.3,-2.5);
    \draw [dotted] (8,0) -- (8,-2);
    \draw [dotted] (5.9,0) -- (5.9,-3.5);
    \draw [dotted] (6.3,0) -- (6.3,-3);

    \path (-1.5,-3.2) node [left] {$\vdots$};
  \end{tikzpicture}
  \caption{A partial cartography of our problem}
  \label{fig-carto}
\end{figure}
The picture seems rather complete, since only the status of
$\problem_{\calM,\nu_1,\nu_2}(\overline\gamma)$ remains
uncertain. However, it remains to discuss two things: first,
the~limits $\underline\gamma$ and $\overline\gamma$ are a priori
unknown, hence the cartography is not effective so~far. The~idea is
then to use the sequences $(\underline\val_N)_N$ and
$(\overline\val_N)_N$ to approximate the limits. We will therefore
discuss cases where the two limits coincide (we~then say that the
approach is \emph{almost-complete}), allowing for a converging scheme
and hence an algorithm to almost cover the interval $[0,1]$ with
either red (\emph{there are no solutions}) or green (\emph{there is a
  solution}), that is, to almost compute the full cartography of the
problem. Second, we~should discuss the effectiveness of the approach.

\section{Almost-completeness of the approach}

In this section, we discuss the almost-completeness of our approach,
and describe situations where one can show that $\underline\gamma =
\overline\gamma \stackrel{\text{def}}{=} \gamma$, which allows to
reduce the unknown part of the cartography to the
singleton~$\{\gamma\}$.

The situations for completeness we describe below are conditions over
cycles, either on weight~$w_1$ or on weight~$w_2$.
When we assume that cycles have
positive $w_i$-weights, we mean it for every cycle, except for cycles
containing~$\Goal$, which we assumed are self-loops with weight~$0$.

\subsection{When all cycles have a positive $w_2$-weight}

In this subsection, we assume that the $w_2$-weight of each cycle
of~$\calM$ is positive (this~is the case for instance when the
$w_2$-weight of each edge is~$1$, i.e., when $w_2$ counts the number
of steps). We~let~$n$ be the number of states of~$\calM$.

\begin{lemma}
  \label{lemma:count}
  There exists a constant~$\kappa\geq 0$ such that, 
  for any strategy~$\sigma$ satisfying
  \(
  \Esp_\sigma(\TS_{w_2}^\Goal) < \nu_2
  \),
  and any~$N>n$, it~holds:
  \[
  \Prob_\sigma\Big(\phi_N^- \vee \phi_N^+\Big) \geq
  1-\frac{n}{N-n}\cdot\kappa.
  \]
\end{lemma}

\begin{proof}
  Assuming otherwise, the impact of all runs that do not belong to
  $\phi_N^- \vee \phi_N^+$ would be too large for the constraint
  on~$\TS_{w_2}^\Goal$.  Indeed, applying the law of total
  expectation, we~can write for every~$N>n$:
  \begin{xalignat*}1
    \Esp_\sigma^\calM\Big(\TS_{w_2}^\Goal\Big) & = 
    \Esp_\sigma^\calM\Big(\TS_{w_2}^\Goal \mid \F[\le N] \Goal \Big)
    \cdot \Prob_\sigma^\calM\Big(\F[\le N] \Goal\Big) +  
    \Esp_\sigma^\calM\Big(\TS_{w_2}^\Goal \mid \G[\le N] \neg \Goal \Big)
    \cdot \Prob_\sigma^\calM\Big(\G[\le N] \neg \Goal\Big)
  \end{xalignat*}
  Write $W_2$ for the minimal (possibly negative) $w_2$-weight
  appearing in~$\calM$, and $c_2$ for the minimal (positive by
  hypothesis) $w_2$-weight of cycles in~$\calM$.
  Noticing that, along any path,
  at most $n$ edges may be outside any cycle, we~get
  \[
  \Esp_\sigma^\calM\Big(\TS_{w_2}^\Goal \mid \F[\le N] \Goal \Big) \geq n\cdot W_2
  \]
  and
  \[
  \Esp_\sigma^\calM\Big(\TS_{w_2}^\Goal \mid \G[\le N] \neg \Goal \Big)
  \geq n\cdot W_2 + \frac {N-n}n \cdot c_2.
  \]
  We~get:
  \begin{xalignat*}1
  \Esp_\sigma^\calM\Big(\TS_{w_2}^\Goal\Big) & \geq
  n\cdot W_2\cdot \Bigl(\Prob_\sigma^\calM\Big(\F[\le N] \Goal\Big)+
  \Prob_\sigma^\calM\Big(\G[\le N] \neg \Goal\Big)\Bigr) +
  \frac{N-n}{n}\cdot c_2  \cdot 
  \Prob_\sigma^\calM\Big(\G[\le N] \neg \Goal\Big)
  \end{xalignat*}
  Since the left-hand side is strictly smaller than~$\nu_2$,
  we~get
  \[
  \Prob_\sigma^\calM\Bigl(\G[\le N] \neg \Goal\Bigr) =
  \Prob_\sigma^\calM\Bigl(\psi_N\Bigr) = 1-\Prob_\sigma^\calM\Bigl(\phi^+_N\vee \phi^-_N\Bigr)
  \leq \frac {n}{N-n}\cdot \Bigl[\frac{\nu_2-W_2\cdot n}{c_2}\Bigr].
  \]
\end{proof}

\begin{lemma}
  For any constant~$\kappa$ satisfying Lemma~\ref{lemma:count},
  and any $N>n$,  we~have
  \[
  0 \le \overline\val_N -\underline\val_N \le \frac{n}{N-n} \cdot \kappa.
  \]
\end{lemma}

\begin{proof}
  We~already remarked that $\underline\val_N \le
  \overline\val_N$.  Now, from Lemma~\ref{lemma:count}, for
  every strategy~$\sigma$ such that
  $\Esp_\sigma\Big(\TS_{w_2}^\Goal\Big) < \nu_2$, it~holds for
  any $N>n$ that
  \[
  \Prob_\sigma\Big(\psi_N\Big) < \frac{n}{N-n}\cdot\kappa.
  \]
  Hence, for any strategy~$\sigma$ and any $N>n$,
  \[
    \Prob_\sigma \Big( \phi_N^- \vee \psi_N\Big)  =
    \Prob_\sigma \Big( \phi_N^-\Big) +
    \Prob_\sigma \Big( \psi_N\Big) 
    \leq
    \Prob_\sigma\Big(\phi_N^-\Big) +\frac{n}{N-n}\cdot\kappa.
  \]
  Taking the infimum over $\sigma$, first in the left-hand side, and then
  in the right-hand side, we~get the expected bound.
\end{proof}

\begin{corollary}
  $\underline\gamma= \overline\gamma$.
\end{corollary}

\begin{remark}
  Notice that the result does not hold without the assumption.
  Indeed, consider
  the 2w-MDP defined by the two deterministic edges $s
  \xrightarrow{a,0,0} s$ and $s \xrightarrow{b,-1,0} \Goal$, with
  $\nu_1=0$.  Then, for every $N$, $\underline\val_N = 0$, while
  $\overline\val_N=1$.
\end{remark}

\subsection{When all cycles have a positive $w_1$-weight}

We assume that each cycle of $\calM$ has a positive $w_1$-weight. We
first notice that:
\begin{lemma}\label{lemma-N0}
  There exists an integer $N_0$ such that for every path $\rho$ from
  $s_\init$ of length ${N \ge N_0}$ not visiting the goal state, it~holds
  \(
  \Acc_{w_1}^N(\rho) 
  \ge \nu_1
  \).
  In particular, if $\rho$ satisfies $\F[\ge N_0] (\neg \Goal \wedge
  \F \Goal)$, then $\TS_{w_1}^\Goal(\rho) \ge \nu_1$.
\end{lemma}

Using this remark, we can prove:
\begin{lemma}
  $\underline\gamma= \overline\gamma$
\end{lemma}

\begin{proof}
  We fix the index $N_0$ as in Lemma~\ref{lemma-N0}.
  For~any~$N \ge N_0$ and any path~$\rho$ of length larger than~$N$,
  we~have
  \[
  \rho \models \phi_N^- \iff \rho \models \phi_{N_0}^-\qquad \text{and}
  \qquad \rho \models \phi_N^+ \iff \rho \models \phi_{N_0}^+ \vee
  \Big( \G[\le N_0] \neg \Goal \wedge
  \F[{(N_0; N]}] \Goal \Big).
  \]
  From the first equivalence, we infer that for every $N \ge N_0$,
  $\underline\val_N = \underline\val_{N_0}$.
  
  Let $N> N_0$, and write:
  \begin{xalignat*}1 \overline\val_N & = \inf \Big\{\Prob_\sigma
    \Big(\phi_{N}^- \vee \psi_N\Big) \mid \sigma\ \text{s.t.}\
    \Esp_{\sigma}\Big(\TS_{w_2}^\Goal\Big) <
    \nu_2\Big\} \\
    &= 1-\sup\{\Prob_\sigma \Big(\phi_{N}^+\Big) \mid \sigma\
    \text{s.t.}\ \Esp_{\sigma}\Big(\TS_{w_2}^\Goal\Big) <
    \nu_2\Big\} \\
    &= 1-\sup\{\Prob_\sigma \Big(\phi_{N_0}^+ \vee \Big( \G[\le N_0]
    \neg \Goal \wedge \F[{(N_0; N]}] \Goal \Big)\Big) \mid \sigma\
    \text{s.t.}\ \Esp_{\sigma}\Big(\TS_{w_2}^\Goal\Big) < \nu_2\Big\}
  \end{xalignat*}
  
  We claim that:
  \begin{restatable}{lemma}{intermlemma}
  \label{lemma:toto}
    \begin{multline*}
      \lim_{N \to +\infty} \sup\{\Prob_\sigma \Big(\phi_{N_0}^+ \vee
      \Big( \G[\le N_0] \neg \Goal \wedge \F[{(N_0; N]}] \Goal
      \Big)\Big) \mid \sigma\ \text{s.t.}\
      \Esp_{\sigma}\Big(\TS_{w_2}^\Goal\Big) <
      \nu_2\Big\} \\
      = \sup\{\Prob_\sigma \Big(\phi_{N_0}^+ \vee \Big(\G[\le N_0]
      \neg \Goal \wedge \F[>N_0] \Goal\Big)\Big) \mid \sigma\
      \text{s.t.}\ \Esp_{\sigma}\Big(\TS_{w_2}^\Goal\Big) <
      \nu_2\Big\}
      % = \sup\{\Prob_\sigma \Big(\phi_{N_0}^+ \vee \neg \Goal \U[>N_0] \Goal\Big) \mid \sigma\
      % \text{s.t.}\ \Esp_{\sigma}\Big(\TS_{w_2}^\Goal\Big) <
      % \nu_2\Big\}
    \end{multline*}
  \end{restatable}

  \ifproof\else\noindent\fi
  From this lemma, we get that:
  \begin{xalignat*}1 \lim_{N \to +\infty} \overline\val_N &= 1 -
    \sup\{\Prob_\sigma \Big(\phi_{N_0}^+ \vee \Big(\G[\le N_0] \neg
    \Goal \wedge \F[>N_0] \Goal\Big)\Big) \mid \sigma\ \text{s.t.}\
    \Esp_{\sigma}\Big(\TS_{w_2}^\Goal\Big) < \nu_2\Big\}
    \\
    & = \inf \{\Prob_\sigma \Big(\phi_{N_0}^- \vee \G \neg \Goal\Big)
    \mid \sigma\ \text{s.t.}\
    \Esp_{\sigma}\Big(\TS_{w_2}^\Goal\Big) < \nu_2\Big\} \\
    \tagcomment{because a path not satisfying
      $\phi_{N_0}^+ \vee \Big(\G[\le N_0]
      \neg \Goal \wedge \F[>N_0] \Goal\Big)$ satisfies $\phi_{N_0}^- \vee \G\neg\Goal$}\\
    & = \inf \{\Prob_\sigma \Big(\phi_{N_0}^- \Big) \mid \sigma\
    \text{s.t.}\
    \Esp_{\sigma}\Big(\TS_{w_2}^\Goal\Big) < \nu_2\Big\} \\
    \tagcomment{since $\Prob_\sigma (\G \neg \Goal) =0$
      when $\Esp_{\sigma}(\TS_{w_2}^\Goal) <  \nu_2$} \\
    & = \underline\val_{N_0}
  \end{xalignat*}
  Hence we conclude that $\underline\gamma = \overline\gamma$ (and
  both limits are reached after finitely many steps).
\end{proof}

\begin{remark}
  This result requires the $w_1$-positivity of cycles, as witnessed by
  the remark at the end of the previous section.

  Also, one could think that assuming $w_1$-negativity of cycles would
  be very similar, but this is not the case, as witnessed by
  the 2w-MDP defined by $s \xrightarrow{a,-1,0} s$ and $s
  \xrightarrow{b,1,0} \Goal$. Then, for every $N \ge 2$,
  $\underline\val_N = 0$ while $\overline\val_N=1$.
\end{remark}

\section{Effectiveness of the approach}

We now explain how the two optimization problems can be
solved. We~first unfold our original 2w-MDP~$\calM$ up to depth~$N$ as
a tree, keeping a copy of~$\calM$ below each leaf; write~$\calT_N$ for
this new 2w-MDP. There~is a natural one-to-one mapping from paths
in~$\calM$ and paths in~$\calT_N$, from which we derive another
one-to-one mapping~$\iota_N$ between strategies in~$\calM$ and
strategies in~$\calT_N$.  Furthermore two corresponding strategies
assign the same probabilities and the same accumulated weights to the
paths. As~a consequence, for any~$N$,
any~$\kappa_N \in \{\phi^+_N,\phi^-_N,\psi_N\}$, and any
strategy~$\sigma$ in~$\calM$, we~have
\[
\Prob^\calM_\sigma \Big(\kappa_N\Big) =
\Prob^{\calT_N}_{\iota_N(\sigma)} \Big(\kappa_N\Big).
\]
We do not formalize the relation between $\calM$ and
$\calT_N$ further, as it~is rather straightforward.

Our two optimization problems can then be rephrased in~$\calT_N$ as
follows:
\[
\overline{\val}_N = \inf \Bigl\{\Prob^{\calT_N}_{\iota_N(\sigma)} \Big(\phi_N^-
\vee \psi_N\Big) \Bigm| \sigma\ \text{s.t.}\
\Esp^{\calT_N}_{\iota_N(\sigma)}\Big(\TS_{w_2}^\Goal\Big) < \nu_2\Bigr\}
\]
and 
\[
\underline{\val}_N = \inf \Bigl\{\Prob^{\calT_N}_{\iota_N(\sigma)}
\Big(\phi_N^-\Big) \Bigm| \sigma\ \text{s.t.}\
\Esp^{\calT_N}_{\iota_N(\sigma)}\Big(\TS_{w_2}^\Goal\Big) < \nu_2\Bigr\}.
\]

From $\calT_N$, we build the finite tree~$\widehat\calT_N$ as~follows:
we~keep the first~$N$ levels of~$\calT_N$, add a fresh state~$\Goal$
(at~level~$N+1$), and from each leaf at level~$N$, corresponding to
some state $s$ of~$\calM$, we~add an edge to~$\Goal$ labelled by the
$w_2$-stochastic shortest path value from~$s$ in~$\calM$, that is,
$\inf_\sigma \{\Esp_{\sigma,s}^\calM (\TS_{w_2}^\Goal)\}$.  Those can
be computed~\cite{BBD+18} (note~that each can either be~$-\infty$ or a
finite value, or~$+\infty$ if~$\Goal$ cannot be almost-surely
reached).

Every strategy $\sigma_N$ in~$\calT_N$ can then be partly mimicked
in~$\widehat\calT_N$ (up~to the $N$-th level of the~tree);
at~level~$N$, there is a single transition, which directly
reaches~$\Goal$ while increasing weight~$w_2$ by the shortest-path
value mentioned above.  We~write $\widehat\sigma_N$ for this strategy
in~$\widehat\calT_N$ derived from~$\sigma_N$.  Then we~have
$\Esp^{\widehat\calT_N}_{\widehat\sigma_N}(\TS_{w_2}^\Goal) \le
\Esp^{\calT_N}_{\sigma_N}(\TS_{w_2}^\Goal)$, since the transitions
from the nodes at level~$N$ to the~$\Goal$ state in~$\widehat\calT_N$
somehow acts as an``optimal'' strategy after the first $N$ levels.

Conversely, for every strategy~$\widehat\sigma_N$ in~$\widehat\calT_N$
such that
$\Esp^{\widehat\calT_N}_{\widehat\sigma_N}(\TS_{w_2}^\Goal) < \nu_2$:
\begin{itemize}
\item if
  $\Esp^{\widehat\calT_N}_{\widehat\sigma_N}(\TS_{w_2}^\Goal) =
  -\infty$, then for every $r \in \mathbb{R}$, one can extend
  $\widehat\sigma_N$ into a strategy~$\sigma_{N,r}$ in~$\calT_N$ such that
  $\Esp^{\calT_N}_{\sigma_{N,r}}(\TS_{w_2}^\Goal) < r$;
\item otherwise, if
  $\Esp^{\widehat\calT_N}_{\widehat\sigma_N}(\TS_{w_2}^\Goal) \ne
  -\infty$, then one can extend $\widehat\sigma_N$ into a strategy
  $\sigma_N$ in~$\calT_N$ such that
  $\Esp^{\widehat\calT_N}_{\widehat\sigma_N}(\TS_{w_2}^\Goal) =
  \Esp^{\calT_N}_{\sigma_N}(\TS_{w_2}^\Goal)$.
\end{itemize}
Hence the set of strategies $\widehat\sigma_N$ in $\widehat\calT_N$ such
that $\Esp^{\widehat\calT_N}_{\widehat\sigma_N}(\TS_{w_2}^\Goal)
< \nu_2$ coincides with the set of strategies obtained as a pruning of
a strategy $\sigma_N$ in $\calT_N$ such that
$\Esp^{\calT_N}_{\sigma_N}(\TS_{w_2}^\Goal) < \nu_2$.
Our~two optimization problems can then be rephrased~as:
\[
\overline{\val}_N = \inf
\Big\{\Prob^{\widehat\calT_N}_{\widehat\sigma_N} \Big(\phi_N^- \vee
\psi_N\Big) \mid \widehat\sigma_N\ \text{s.t.}\
\Esp^{\widehat\calT_N}_{\widehat\sigma_N}\Big(\TS_{w_2}^\Goal\Big) <
\nu_2\Big\}
\]
and 
\[
\underline{\val}_N = \inf
\Big\{\Prob^{\widehat\calT_N}_{\widehat\sigma_N} \Big(\phi_N^-\Big) \mid
\widehat\sigma_N\ \text{s.t.}\
\Esp^{\widehat\calT_N}_{\widehat\sigma_N}\Big(\TS_{w_2}^\Goal\Big) <
\nu_2\Big\}
\]

Since $\widehat\calT_N$ is a (finite) tree, each
strategy~$\widehat\sigma_N$ in that MDP is memoryless, and can be
represented as a probability value given to each edge appearing in the
tree.

For each node $\mathsf{n}$ of $\widehat\calT_N$, corresponding to some
state~$s$ of~$\calM$, and for each edge $e = (s,\delta)$ from~$s$,
we~consider a variable~$p_{\mathsf{n},e}$, intended to represent the
probability of taking edge~$e$ at node~$\mathsf{n}$. In~particular, we
will have the constraints $0 \le p_{\mathsf{n},e} \le 1$ and
$\sum_{e=(s,\delta)} p_{\mathsf{n},e} = 1$. We write $\mathfrak{P} =
(p_{\mathsf{n},e})_{\mathsf{n},e}$ for the tuple of all variables.

The two optimization problems above can then be written as:
\[
\inf_{\mathfrak{P}}\Big\{P(\mathfrak{P}) \Bigm| Q(\mathfrak{P}) < \nu_2
\wedge \bigwedge_{\mathsf{n},e} 0 \le p_{\mathsf{n},e} \le 1 \wedge
\bigwedge_{\mathsf{n}} \sum_{e=(s,\delta)} p_{\mathsf{n},e} = 1\Big\}
\]
where $P(\mathfrak{P})$ and $Q(\mathfrak{P})$ are polynomials (of
degree at most $N$).

Such polynomial optimization problems are in general hard to solve,
and we have not been able to exploit the particular shape of our
optimization problem to get efficient specialized
algorithms. For~each~$N$, arbitrary under-approximations
of~$\underline\val_N$ and over-approximations of~$\overline\val_N$ can
be obtained by binary search, using the existential theory of the
reals: the~latter problem can be solved in polynomial space, but the
number of variables of our problem is exponential in~$N$. Using
Lemmas~\ref{lemma:sol} and~\ref{lemma:underline}, we~get informations
about the cartography of our problem. We~can get approximations
of~$\underline\gamma$ and~$\overline\gamma$ by iterating this
procedure for larger values of~$N$.

\begin{remark}
  In case the constraint on the expectation of $w_2$ can be relaxed
  (for~instance if it is trivially satisfied), then the problem
  (over~$\widehat\calT_N$) becomes a simple optimal reachability
  problem in an~MDP, for~which pure strategies are sufficient (we~have
  seen that this cannot be the case in our setting). The above
  optimization problem then simplifies into a linear-programming problem,
  with a much better complexity.
\end{remark}

\section{The special case of $\problem_{\calM,\nu_1,\nu_2}(0)$}

While the previous developments cannot give a solution to
$\problem_{\calM,\nu_1,\nu_2}(0)$ since it requires not only to show
that $\underline\gamma = \overline\gamma=0$, but also that there is a
solution to the limit point $\overline\gamma$ (which we do not have in
general).  We dedicate special developments to that problem. In~this
section, we~assume (w.l.o.g.) that weights take integer values.

$\problem_{\calM,\nu_1,\nu_2}(0)$ can be rephrased as follows: \emph{there
exists a strategy $\sigma_0$ such that:
\begin{enumerate}
\item for all $\rho \in \out^\calM_{\infty}(\sigma_0,s_\init)$,
  $\TS_{w_1}^{\Goal}(\rho) \ge \nu_1$;
\item $\Esp^\calM_{\sigma_\epsilon,s_\init}\Big(\TS_{w_2}^\Goal\Big) <
  \nu_2$.
\end{enumerate}}
Note that this problem is somehow a ``beyond worst-case problem'', as
defined in~\cite{RRS17}, with a strong constraint on all outcomes, and
a stochastic constraint (here defined using expected value).

We describe a solution in the case all cycles of~$\calM$ have
non-negative $w_1$-weights, which is inspired
from~\cite[Theorem~13]{RRS17}.  As~we explain below, our~solution
extends to multiple weights (with non-negative cycles) with strong
constraints (like the one for~$w_1$). However, it~is not correct when
$w_1$ may have negative cycles as well. In that case, the status of
the problem remains open.

We~``unfold'' the 2w-MDP~$\calM = (S,s_\init,\Goal,E,w_1,w_2)$ into a
1w-MDP $\calN = (Q,q_\init,Q_\Goal,T,w)$, explicitly keeping track
of~$w_1$ in the states of~$\calN$:
\begin{itemize}
\item
  $Q = S \times \{-M, -M+1, \ldots, 0,1,\ldots,M+\lfloor
  \nu_1\rfloor+1,\infty\}$, where $M = W \cdot (|S|+1)$, $W$~is the
  maximal absolute value of all weights~$w_1(s,s')$ in~$\calM$, and
  $\lfloor \nu_1\rfloor$ is the integral part of $\nu_1$;
\item $q_\init = (s_\init,0)$;
\item $Q_\Goal = \{(\Goal,k) \mid k =\infty\ \text{or}\ k \ge
  \nu_1\}$;
\item $T = \Bigl\{((s,c),(s',c')) \Bigm| (s,s') \in E\ \text{and}\
  \begin{array}{l}
    c'=\infty    \qquad \text{ if $c+w_1(s,s')>M+\lfloor \nu_1 \rfloor
    + 1$ or $c=\infty$} \\
    c'= c+w_1(s,s') \qquad\text{otherwise}
  \end{array}\Bigr\}$;
\item $w ((s,c),(s',c')) = w_2(s,s')$.
\end{itemize}
There is a natural one-to-one correspondence $\lambda$ between paths
in~$\calM$ and those in~$\calN$:
$\lambda(s_0s_1 \dots s_k \dots) = (s_0,0)(s_1,c_1) \dots (s_k,c_k)
\dots$ where, for every $k$, $c_k=c_{k-1}+w_1(s_{k-1},s_k)$ if this
value is less than or equal to~$M + \lfloor \nu_1 \rfloor +1$, and
$c_k=\infty$ otherwise. Notice that thanks to our hypothesis on
cycles, $c_k$~may never be less than~$-M$.

Also, by construction, if $(s_0,0)(s_1,c_1) \dots (s_k,c_k) \dots$ is
a path in~$\calN$ such that $c_k=\infty$, then for every $j \ge k$,
$c_j=\infty$; in~that case, in~the corresponding path
$s_0 s_1 \dots s_k \dots$ in~$\calM$, it~is the case for every
$j \ge k$ that $\Acc_{w_1}^j(s_0 s_1 \dots s_k \dots) \ge \nu_1$
(indeed, once the accumulated weight has become larger than $M+\nu_1$,
it can never be smaller than $\nu_1$ again, thanks to the hypothesis
on cycles). Conversely, if $c_k<\infty$, then
$\Acc^k_{w_1}(s_0 s_1 \dots s_k \dots) =c_k$.

From that correspondence over paths, strategies in~$\calM$ can
equivalently be seen as strategies in~$\calN$ via a mapping~$\iota$.
Using this correspondence:

\begin{restatable}{lemma}{dernierlemme}
  There is a solution to $\problem_{\calM,\nu_1,\nu_2}(0)$ if, and
  only~if, in~$\calN$ there is a strategy~$\tau$ such that:
  \begin{enumerate}
  \item for all $\rho \in \out_\infty^\calN(\tau,q_\init)$, $\rho
    \models \F Q_\Goal$;
  \item $\Esp^\calN_{\tau,q_\init}\Big(\TS_{w}^{Q_\Goal}\Big) <
    \nu_2$.
  \end{enumerate}
  Furthermore, if $\sigma$ is a solution to
  $\problem_{\calM,\nu_1,\nu_2}(0)$, then $\tau=\iota(\sigma)$ is a solution
  to the above problem in~$\calN$; and if $\tau$ is a solution to the
  above problem in~$\calN$, then $\sigma=\iota^{-1}(\tau)$ is a solution to
  $\problem_{\calM,\nu_1,\nu_2}(0)$.
\end{restatable}

\ifproof
\begin{proof}
  Let $\sigma$ be a solution to $\problem_{\calM,\nu_1,\nu_2}(0)$, and
  let $\tau=\iota(\sigma)$. Assume towards a contradiction that there
  is $\rho \in \out_\infty^\calN(\tau,q_\init)$ such that $\rho
  \not\models \F Q_\Goal$. Write $\rho = (s_0,0) (s_1,c_1) \ldots
  (s_k,c_k) \ldots$ By~correspondence of the strategies,
  the~path~$\lambda^{-1}(\rho)$ belongs to $
  \out_\infty^\calM(\sigma,s_\init)$, hence $\lambda^{-1}(\rho)
  \models \TS_{w_1}^\Goal \ge \nu_1$ and $\rho \models \F (\{\Goal\}
  \times \{0,1,\ldots,M,\infty\})$. However, since $\rho \not\models \F
  Q_\Goal$, $\rho$~must end up in  some state~$(\Goal,c) \notin
  Q_\Goal$, hence with ${c<\nu_1}$. Justified by the previous
  discussion, $\rho$~never visits a state~$(s,\infty)$, hence for
  every $k$, $c_k = \Acc_{w_1}^k(\lambda^{-1}(\rho))$. This is a
  contradiction.

  Now, the expectation of $w_2$ in $\calM$ is the same as that of $w$
  in $\calN$, hence the second condition is also satisfied. Strategy
  $\tau$ is therefore a solution to the above problem.

  \smallskip

  Conversely, pick a strategy $\tau$ in $\calN$ which is
  solution to the problem described in the Lemma.
  Let~$\sigma=\iota^{-1}(\tau)$.
  Since any path $\rho \in
  \out_\infty^\calN(\tau,q_\init)$ satisfies $\F Q_\Goal$, we get
  that any path $\rho \in \out_\infty^\calM(\sigma,s_\init)$ satisfies
  $\F \Goal$ and $\Prob^\calM_{\sigma} (\TS_{w_1}^\Goal \ge
  \nu_1) = 1$ (this is easy to prove). The constraint on the
  expectation of $w_2$ is also straightforward. This concludes the
  proof.
\end{proof}
\fi

\ifproof\else We can further show that
$\problem_{\calM,\nu_1,\nu_2}(0)$ has a solution if, and only~if, the
stochastic $w_2$-shortest-path
% -value~$\pi$
for reaching~$Q_\Goal$ in~$\calN$ from~$s_\init$ is smaller
than~$\nu_2$.
This latter problem can be decided in \PTIME~\cite{BBD+18}. However,
the~size of~$\calN$ is exponential (more precisely, it~is
pseudo-polynomial in the size of~$\calM$).
In~the~end:
\fi
\begin{restatable}{theorem}{thmprobzero}
  One can decide $\problem_{\calM,\nu_1,\nu_2}(0)$ in
  pseudo-polynomial time, when cycles of $\calM$ have a non-negative
  $w_1$-weight.
\end{restatable}

\begin{remark}
  Notice that we could also have assumed that all cycles have
  non-positive $w_1$-weight: the~construction of~$\calN$ would be
  similar, but with state space $S\times
  \{-\infty,\lfloor\nu_1\rfloor-1-M', ..., 0, 1, ..., M'+1\}$ where
  $M'=W'\cdot (|S|+1)$ and $W'$~is the largest absolute value $w_1$-weight
  in~$\calM$. The~rest of the argumentation follows the same ideas as above.

  Notice also that our algorithm is readily adapted to the case where
  we may have several constraints similar to those on~$w_1$ (for each
  extra variable, one should assume that either each cycle is
  non-negative, or each cycle is non-positive). It~suffices to keep
  track of the extra weights in the states, and take as target all
  states where the constraints are fulfilled.
\end{remark}

\section{Conclusion}

In this paper, we investigated a multi-constrained reachability
problem over MDPs, which originated in the context of electric-vehicle
charging~\cite{GBBLM17}.
This~problem consists in finding a strategy that surely reaches a
quantitative goal (e.g., \emph{all vehicles are fully charged and the
  load of the network remains below a given bound at any time}) while
satisfying a condition on the expected value of some variable
(\emph{the~life expectancy of the transformer is high} or \emph{the
  expected cost of charging all vehicles is minimized}).
We~developed partial solutions to the problem by providing a
cartography of the solutions to (a~relaxed version~of) the problem.
We~identified realistic conditions under which the cartography is
(almost) complete. However, even under these conditions, the general decision
problem (given $\epsilon$, does $\problem(\epsilon)$ have a solution?)
remains open so far.
Also, the~case of MDPs not satisfying these conditions
%is unclear
remains also open, but we believe that our approximation techniques
may give interesting informations which suffice for practical
applications such as electric-vehicle charging.

Our approach for $\problem(0)$, which amounts to explicitly keep track
of the worst-case constraint on~$w_1$, immediately extends to multiple
weights with worst-case constraints (with~the same assumptions on
cycles---note~that the more general setting could not be solved, which
has to be put in parallel with the undecidability result
of~\cite[Theorem 12]{RRS17}) for the multi-dimensional percentile
problem for truncated sum payoffs.
% {\color{red}\fbox{expliciter ce r\'esultat, je ne l'ai pas sous la
% main}}
%
The~cartography for the relaxed problem $\problem(\epsilon)$ requires
solving sequences of intermediary optimization problems, which can be
expressed as polynomial optimization problems with polynomial
constraints.  It~could be extended to several such weights as~well
(either~by putting an assumption on the $w_1$-weights of every cycle,
or just on the $w_2$-weight of every cycle).
A~nice continuation of our work would consist in computing
(approximations~of) Pareto-optimal solutions in such a
setting. Improving the complexity and practicality of our approach is
also on our agenda for future work.

\paragraph{Acknowledgement.}
We thank the anonymous reviewers for their careful reading of our
submission.

Patricia Bouyer, Mauricio Gonz{\'a}lez and Nicolas Markey are
supported by ERC project EQualIS. Mickael Randour is an F.R.S.-FNRS
Research Associate, and he is supported by the F.R.S.-FNRS Incentive
Grant ManySynth.

\makeatletter
\providecommand{\doi}{\catcode`\_11\@doi}
\providecommand{\@doi}[1]{doi:\urlalt{http://dx.doi.org/#1}{#1}}
%% slightly shortened \itemsep
\renewenvironment{thebibliography}[1]
     {\section*{\refname}\small%              small added
      \list{\@biblabel{\@arabic\c@enumiv}}%
           {\settowidth\labelwidth{\@biblabel{#1}}%
            \leftmargin\labelwidth
            \advance\leftmargin\labelsep
            \@openbib@code
            \usecounter{enumiv}%
            \let\p@enumiv\@empty
            \renewcommand\theenumiv{\@arabic\c@enumiv}}%
      \sloppy
      \clubpenalty4000
      \@clubpenalty \clubpenalty
      \widowpenalty4000%
      \sfcode`\.\@m
                \setlength{\parskip}{0pt}%
                \setlength{\itemsep}{3pt plus 2pt minus 1pt}% less space between items
     }
     {\def\@noitemerr
       {\@latex@warning{Empty `thebibliography' environment}}%
      \endlist}
\makeatother
\bibliographystyle{eptcs}
%\bibliography{biblio,extra}
\bibliography{bibexport}

\begin{thebibliography}{10}
\providecommand{\bibitemdeclare}[2]{}
\providecommand{\surnamestart}{}
\providecommand{\surnameend}{}
\providecommand{\urlprefix}{Available at }
\providecommand{\url}[1]{\texttt{#1}}
\providecommand{\href}[2]{\texttt{#2}}
\providecommand{\urlalt}[2]{\href{#1}{#2}}
\providecommand{\doi}[1]{doi:\urlalt{http://dx.doi.org/#1}{#1}}
\providecommand{\bibinfo}[2]{#2}

\bibitemdeclare{inproceedings}{BBD+18}
\bibitem{BBD+18}
\bibinfo{author}{Christel \surnamestart Baier\surnameend},
  \bibinfo{author}{Nathalie \surnamestart Bertrand\surnameend},
  \bibinfo{author}{Clemens \surnamestart Dubslaff\surnameend},
  \bibinfo{author}{Daniel \surnamestart Gburek\surnameend} \&
  \bibinfo{author}{Ocan \surnamestart Sankur\surnameend}
  (\bibinfo{year}{2018}): \emph{\bibinfo{title}{Stochastic Shortest Paths and
  Weight-Bounded Properties in {M}arkov Decision Processes}}.
\newblock In: {\sl \bibinfo{booktitle}{LICS'18}}, \bibinfo{publisher}{IEEE},
  \doi{10.1145/3209108.3209184}.

\bibitemdeclare{inproceedings}{BDK14}
\bibitem{BDK14}
\bibinfo{author}{Christel \surnamestart Baier\surnameend},
  \bibinfo{author}{Clemens \surnamestart Dubslaff\surnameend} \&
  \bibinfo{author}{Sascha \surnamestart Kl{\"u}ppelholz\surnameend}
  (\bibinfo{year}{2014}): \emph{\bibinfo{title}{Trade-off analysis meets
  probabilistic model checking}}.
\newblock In: {\sl \bibinfo{booktitle}{CSL-LICS'14}}, \bibinfo{publisher}{ACM},
  pp. \bibinfo{pages}{1:1--1:10}, \doi{10.1145/2603088.2603089}.

\bibitemdeclare{inproceedings}{BKKW17}
\bibitem{BKKW17}
\bibinfo{author}{Christel \surnamestart Baier\surnameend},
  \bibinfo{author}{Joachim \surnamestart Klein\surnameend},
  \bibinfo{author}{Sascha \surnamestart Kl{\"u}ppelholz\surnameend} \&
  \bibinfo{author}{Sascha \surnamestart Wunderlich\surnameend}
  (\bibinfo{year}{2017}): \emph{\bibinfo{title}{Maximizing the Conditional
  Expected Reward for Reaching the Goal}}.
\newblock In: {\sl \bibinfo{booktitle}{TACAS'17}}, {\sl \bibinfo{series}{LNCS}}
  \bibinfo{volume}{10206}, \bibinfo{publisher}{Springer}, pp.
  \bibinfo{pages}{269--285}, \doi{10.1007/978-3-662-54580-5_16}.

\bibitemdeclare{article}{BLHM16}
\bibitem{BLHM16}
\bibinfo{author}{Olivier \surnamestart Beaude\surnameend},
  \bibinfo{author}{Samson \surnamestart Lasaulce\surnameend},
  \bibinfo{author}{Martin \surnamestart Hennebel\surnameend} \&
  \bibinfo{author}{Ibrahim \surnamestart Mohand-Kaci\surnameend}
  (\bibinfo{year}{2016}): \emph{\bibinfo{title}{Reducing the Impact of {EV}
  Charging Operations on the Distribution Network}}.
\newblock {\sl \bibinfo{journal}{IEEE Trans.\ Smart Grid}}
  \bibinfo{volume}{7}(\bibinfo{number}{6}), pp. \bibinfo{pages}{2666--2679},
  \doi{10.1109/TSG.2015.2489564}.

\bibitemdeclare{inproceedings}{BRR17}
\bibitem{BRR17}
\bibinfo{author}{Rapha{\"e}l \surnamestart Berthon\surnameend},
  \bibinfo{author}{Mickael \surnamestart Randour\surnameend} \&
  \bibinfo{author}{Jean-Fran\c{c}ois \surnamestart Raskin\surnameend}
  (\bibinfo{year}{2017}): \emph{\bibinfo{title}{Threshold Constraints with
  Guarantees for Parity Objectives in {M}arkov Decision Processes}}.
\newblock In: {\sl \bibinfo{booktitle}{ICALP'17}}, {\sl
  \bibinfo{series}{LIPIcs}}~\bibinfo{volume}{80}, \bibinfo{publisher}{LZI}, pp.
  \bibinfo{pages}{121:1--121:15}, \doi{10.4230/LIPIcs.ICALP.2017.121}.

\bibitemdeclare{article}{BBC+14}
\bibitem{BBC+14}
\bibinfo{author}{Tom\'a\v{s} \surnamestart Br{\'a}zdil\surnameend},
  \bibinfo{author}{V{\'a}clav \surnamestart Bro\v{z}ek\surnameend},
  \bibinfo{author}{Krishnendu \surnamestart Chatterjee\surnameend},
  \bibinfo{author}{Vojt\v{e}ch \surnamestart Forejt\surnameend} \&
  \bibinfo{author}{Anton{\'i}n \surnamestart Ku\v{c}era\surnameend}
  (\bibinfo{year}{2014}): \emph{\bibinfo{title}{Markov Decision Processes with
  Multiple Long-Run Average Objectives}}.
\newblock {\sl \bibinfo{journal}{LMCS}}
  \bibinfo{volume}{10}(\bibinfo{number}{1:13}), pp. \bibinfo{pages}{1--29},
  \doi{10.2168/LMCS-10(1:13)2014}.

\bibitemdeclare{article}{BFRR17}
\bibitem{BFRR17}
\bibinfo{author}{V{\'e}ronique \surnamestart Bruy{\`e}re\surnameend},
  \bibinfo{author}{Emmanuel \surnamestart Filiot\surnameend},
  \bibinfo{author}{Mickael \surnamestart Randour\surnameend} \&
  \bibinfo{author}{Jean-Fran\c{c}ois \surnamestart Raskin\surnameend}
  (\bibinfo{year}{2017}): \emph{\bibinfo{title}{Meet your expectations with
  guarantees: {B}eyond worst-case synthesis in quantitative games}}.
\newblock {\sl \bibinfo{journal}{Inf.~\& Comp.}} \bibinfo{volume}{254}, pp.
  \bibinfo{pages}{259--295}, \doi{10.1016/j.ic.2016.10.011}.

\bibitemdeclare{inproceedings}{chatterjee07b}
\bibitem{chatterjee07b}
\bibinfo{author}{Krishnendu \surnamestart Chatterjee\surnameend}
  (\bibinfo{year}{2007}): \emph{\bibinfo{title}{Markov Decision Processes with
  Multiple Long-Run Average Objectives}}.
\newblock In: {\sl \bibinfo{booktitle}{FSTTCS'07}}, {\sl
  \bibinfo{series}{LNCS}} \bibinfo{volume}{4855},
  \bibinfo{publisher}{Springer}, pp. \bibinfo{pages}{473--484},
  \doi{10.1007/978-3-540-77050-3_39}.

\bibitemdeclare{article}{CKK17}
\bibitem{CKK17}
\bibinfo{author}{Krishnendu \surnamestart Chatterjee\surnameend},
  \bibinfo{author}{Zuzana \surnamestart K\v{r}et{\'i}nsk{\'a}\surnameend} \&
  \bibinfo{author}{Jan \surnamestart K\v{r}et{\'i}nsk{\'y}\surnameend}
  (\bibinfo{year}{2017}): \emph{\bibinfo{title}{Unifying two views on multiple
  mean-payoff objectives in {M}arkov decision processes}}.
\newblock {\sl \bibinfo{journal}{LMCS}}
  \bibinfo{volume}{13}(\bibinfo{number}{2:15}), pp. \bibinfo{pages}{1--50},
  \doi{10.23638/LMCS-13(2:15)2017}.

\bibitemdeclare{inproceedings}{CR15}
\bibitem{CR15}
\bibinfo{author}{Lorenzo \surnamestart Clemente\surnameend} \&
  \bibinfo{author}{Jean-Fran\c{c}ois \surnamestart Raskin\surnameend}
  (\bibinfo{year}{2015}): \emph{\bibinfo{title}{Multidimensional beyond
  worst-case and almost-sure problems for mean-payoff objectives}}.
\newblock In: {\sl \bibinfo{booktitle}{LICS'15}}, \bibinfo{publisher}{IEEE},
  pp. \bibinfo{pages}{257--268}, \doi{10.1109/LICS.2015.33}.

\bibitemdeclare{article}{DI14}
\bibitem{DI14}
\bibinfo{author}{Jonathan \surnamestart Donadee\surnameend} \&
  \bibinfo{author}{Marija~D. \surnamestart Ilic\surnameend}
  (\bibinfo{year}{2014}): \emph{\bibinfo{title}{Stochastic Optimization of Grid
  to Vehicle Frequency Regulation Capacity Bids}}.
\newblock {\sl \bibinfo{journal}{{IEEE} Trans.\ on Smart Grid}}
  \bibinfo{volume}{5}(\bibinfo{number}{2}), pp. \bibinfo{pages}{1061--1069},
  \doi{10.1109/TSG.2013.2290971}.

\bibitemdeclare{book}{FV97}
\bibitem{FV97}
\bibinfo{author}{Jerzy \surnamestart Filar\surnameend} \& \bibinfo{author}{Koos
  \surnamestart Vrieze\surnameend} (\bibinfo{year}{1997}):
  \emph{\bibinfo{title}{Competitive {M}arkov Decision Processes}}.
\newblock \bibinfo{publisher}{Springer}, \doi{10.1007/978-1-4612-4054-9}.

\bibitemdeclare{inproceedings}{GBBLM17}
\bibitem{GBBLM17}
\bibinfo{author}{Mauricio \surnamestart Gonz{\'a}lez\surnameend},
  \bibinfo{author}{Olivier \surnamestart Beaude\surnameend},
  \bibinfo{author}{Patricia \surnamestart Bouyer\surnameend},
  \bibinfo{author}{Samson \surnamestart Lasaulce\surnameend} \&
  \bibinfo{author}{Nicolas \surnamestart Markey\surnameend}
  (\bibinfo{year}{2017}): \emph{\bibinfo{title}{Strat{\'e}gies d'ordonnancement
  de consommation d'{\'e}nergie en pr{\'e}sence d'information imparfaite de
  pr{\'e}vision}}.
\newblock In: {\sl \bibinfo{booktitle}{GRETSI'17}}.
\newblock \urlprefix\url{http://www.lsv.fr/Publis/PAPERS/PDF/GBBLM-
  gretsi17.pdf}.

\bibitemdeclare{inproceedings}{HK15}
\bibitem{HK15}
\bibinfo{author}{Christoph \surnamestart Haase\surnameend} \&
  \bibinfo{author}{Stefan \surnamestart Kiefer\surnameend}
  (\bibinfo{year}{2015}): \emph{\bibinfo{title}{The Odds of Staying on
  Budget}}.
\newblock In: {\sl \bibinfo{booktitle}{ICALP'15}}, {\sl \bibinfo{series}{LNCS}}
  \bibinfo{volume}{9135}, \bibinfo{publisher}{Springer}, pp.
  \bibinfo{pages}{234--246}, \doi{10.1007/978-3-662-47666-6_19}.

\bibitemdeclare{techreport}{JP16}
\bibitem{JP16}
\bibinfo{author}{Daniel~R. \surnamestart Jiang\surnameend} \&
  \bibinfo{author}{Warren~B. \surnamestart Powell\surnameend}
  (\bibinfo{year}{2016}): \emph{\bibinfo{title}{Practicality of Nested Risk
  Measures for Dynamic Electric Vehicle Charging}}.
\newblock \bibinfo{type}{Research Report} \bibinfo{number}{1605.02848},
  \bibinfo{institution}{arXiv}.

\bibitemdeclare{techreport}{KM18}
\bibitem{KM18}
\bibinfo{author}{Jan \surnamestart Kret{\'i}nsk{\'y}\surnameend} \&
  \bibinfo{author}{Tobias \surnamestart Meggendorfer\surnameend}
  (\bibinfo{year}{2018}): \emph{\bibinfo{title}{Conditional Value-at-Risk for
  Reachability and Mean Payoff in {M}arkov Decision Processes}}.
\newblock \bibinfo{type}{Research Report} \bibinfo{number}{1805.02946},
  \bibinfo{institution}{arXiv}.

\bibitemdeclare{book}{puterman94}
\bibitem{puterman94}
\bibinfo{author}{Martin~L. \surnamestart Puterman\surnameend}
  (\bibinfo{year}{1994}): \emph{\bibinfo{title}{Markov Decision Processes:
  {D}iscrete Stochastic Dynamic Programming}}.
\newblock \bibinfo{publisher}{John Wiley and Sons},
  \doi{10.1002/9780470316887}.

\bibitemdeclare{inproceedings}{RRS15b}
\bibitem{RRS15b}
\bibinfo{author}{Mickael \surnamestart Randour\surnameend},
  \bibinfo{author}{Jean-Fran\c{c}ois \surnamestart Raskin\surnameend} \&
  \bibinfo{author}{Ocan \surnamestart Sankur\surnameend}
  (\bibinfo{year}{2015}): \emph{\bibinfo{title}{Variations on the Stochastic
  Shortest Path Problem}}.
\newblock In: {\sl \bibinfo{booktitle}{VMCAI'15}}, {\sl \bibinfo{series}{LNCS}}
  \bibinfo{volume}{8931}, \bibinfo{publisher}{Springer}, pp.
  \bibinfo{pages}{1--18}, \doi{10.1007/978-3-662-46081-8_1}.

\bibitemdeclare{article}{RRS17}
\bibitem{RRS17}
\bibinfo{author}{Mickael \surnamestart Randour\surnameend},
  \bibinfo{author}{Jean-Fran\c{c}ois \surnamestart Raskin\surnameend} \&
  \bibinfo{author}{Ocan \surnamestart Sankur\surnameend}
  (\bibinfo{year}{2017}): \emph{\bibinfo{title}{Percentile queries in
  multi-dimensional {M}arkov decision processes}}.
\newblock {\sl \bibinfo{journal}{FMSD}}
  \bibinfo{volume}{50}(\bibinfo{number}{2-3}), pp. \bibinfo{pages}{207--248},
  \doi{10.1007/978-3-319-21690-4_8}.

\end{thebibliography}

\end{document}